\documentclass[runningheads,12pt]{llncs}

\usepackage[T1]{fontenc}
\usepackage[utf8]{inputenc}
\usepackage{graphicx}
\usepackage{amssymb}
\usepackage{comment}
\usepackage{mathtools}
\usepackage{tikz}
\usetikzlibrary{shapes.geometric}
\usepackage[ruled,linesnumbered]{algorithm2e}

\usepackage{xspace}

\usepackage{proof-at-the-end}
\pgfkeys{/prAtEnd/sss/.style={
normal}
}
\usepackage{fullpage}

\def\qed{\hfill$\Box$\par\vskip1em}

\begin{document}
\title{Stand-Up Indulgent Gathering on Lines}
\titlerunning{Stand-Up Indulgent Gathering on Lines}

\iftrue

\author{Quentin Bramas\inst{1} \and Sayaka Kamei\inst{2} \and Anissa Lamani\inst{1} \and S\'ebastien Tixeuil \inst{3}}

\institute{
University of Strasbourg, ICube, CNRS, France. 
\and
Graduate School of Advanced Science and Engineering, Hiroshima University, Japan.
\and
Sorbonne University, CNRS, LIP6, IUF, France. }

\else
\institute{}
\fi

\maketitle

\begin{abstract}
We consider a variant of the crash-fault gathering problem called stand-up indulgent gathering (SUIG). In this problem, a group of mobile robots must eventually gather at a single location, which is not known in advance. If no robots crash, they must all meet at the same location. However, if one or more robots crash at a single location, all non-crashed robots must eventually gather at that location. The SUIG problem was first introduced for robots operating in a two-dimensional continuous Euclidean space, with most solutions relying on the ability of robots to move a prescribed (real) distance at each time instant.

In this paper, we investigate the SUIG problem for robots operating in a discrete universe (i.e., a graph) where they can only move one unit of distance (i.e., to an adjacent node) at each time instant. Specifically, we focus on line-shaped networks and characterize the solvability of the SUIG problem for oblivious robots without multiplicity detection.

\end{abstract}

\section{Introduction}

\subsection{Context and Motivation}
Mobile robotic swarms recently received a considerable amount of attention from the Distributed Computing scientific community. Characterizing the exact hypotheses that enable solving basic problems for robots represented as disoriented (each robot has its own coordinate system) oblivious (robots cannot remember past action) dimensionless points evolving in a Euclidean space has been at the core of the researchers' goals for more than two decades.
One of the key such hypotheses is the scheduling assumption~\cite{PGN2019}: robots can execute their protocol fully synchronized (FSYNC), in a completely asynchronous manner (ASYNC), of having repeatedly a fairly chosen subset of robots scheduled for synchronous execution (SSYNC).

Among the many studied problems, the \emph{gathering}~\cite{SY1999} plays a benchmarking role, as its simplicity to express (robots have to gather in finite time at the exact same location, not known beforehand) somewhat contradicts its computational tractability (two robots evolving assuming SSYNC scheduling cannot gather, without additional hypotheses).

As the number of robots grows, the probability that at least one of them fails increases, yet, relatively few works consider the possibility of robot failures.
One of the simplest such failures is the \emph{crash} fault, where a robot unpredictably stops executing its protocol.
In the case of gathering, one should prescribe the expected behavior in the presence of crash failures. 
Two variants have been studied: \emph{weak} gathering expects all correct (that is, non-crashed) robots to gather, regardless of the positions of the crashed robots, while \emph{strong} gathering (also known as stand-up indulgent gathering -- SUIG) expects correct robots to gather at the (supposedly unique) crash location. 
In continuous Euclidean space, weak gathering is solvable in the SSYNC model~\cite{ND2006,ZSS2013,QS2015,XMP2020}, while SUIG (and its variant with two robots, stand up indulgent rendezvous -- SUIR) is only solvable in the FSYNC model~\cite{QAS2020,QAS2021}. 

A recent trend~\cite{PGN2019} has been to move from the continuous environment setting to a discrete one. More precisely, in the discrete setting, robots can occupy a finite number of locations, and move from one location to another if they are neighboring. This neighborhood relation is conveniently represented by a graph whose nodes are locations, leading to the ``robots on graphs'' denomination. This discrete setting is better suited for describing constrained physical environments, or environments where robot positioning is only available from discrete sensors~\cite{TPRLSX19}. From a computational perspective, the continuous setting and the discrete setting are unrelated: on the one hand, the number of possible configurations (that, the number of robot positions) is much more constrained in the discrete setting than in the continuous setting (only a finite number of configurations exists in the discrete setting), on the other hand, the continuous setting offers algorithms designers more flexibility to solve problematic configurations (e.g., using arbitrarily small movements to break a symmetry). 

In this paper, we consider the discrete setting, and aim to characterize the solvability of the SUIR and SUIG problems: in a set of locations whose neighborhood relation is represented by a line-shaped graph, robots have to gather at one single location, not known beforehand; furthermore, if one or more robots crash anytime at the same location, all robots must gather at this location.  
\subsection{Related Works}

In graphs, in the absence of faults, mobile robot gathering was primarily considered for ring-shaped graphs~\cite{REA2008,RAA2010,SAFS2011,SAFS2012,TTSF2013,GGA2014,GAN2017}.
For other topologies, gathering problem was considered, e.g., in 
finite grids~\cite{GGRA2016}, trees~\cite{GGRA2016}, tori~\cite{SAFSK2021}, complete cliques~\cite{SGA2021}, and complete bipartite graphs~\cite{SGA2021}.
Most related to our problem is the (relaxed) FSYNC gathering algorithm presented by Castenow et al.~\cite{CFH2020} for grid-shaped networks where a single robot may be stationary. The main differences with our settings are as follows. First, if no robot is stationary, their~\cite{CFH2020} robots end up in a square of $2\times 2$ rather than a single node as we require. Second, when one robot is stationary (and thus never moves), all other robots gather at the stationary robot location, \emph{assuming a stationary robot can be detected as such when on the same node} (instead, we consider that a crashed robot cannot be detected), and assuming a stationary robot never moves from the beginning of the execution (while we consider anytime crashes). Third, they assume that initial positions are neighboring (while we characterize which patterns of initial positions are solvable).

In the continuous setting, the possibility of a robot failure was previously considered. 
As previously stated, the weak-gathering problem in SSYNC~\cite{ND2006,ZSS2013,QS2015,XMP2020}, and the SUIR and SUIG problems in FSYNC~\cite{QAS2020,QAS2021} were previously considered.
In particular, solutions to SUIR and SUIG~\cite{QAS2020,QAS2021} make use of a level-slicing technique, that mandates them to move by a fraction of the distance to another robot. Obviously, such a technique cannot be translated to the discrete model, where robots always move by exactly one edge.

Works combining the discrete setting and the possibility of robot failures are scarce. 
Ooshita and Tixeuil~\cite{FS2015} considered transient robot faults placing them at arbitrary locations, and presented a probabilistic self-stabilizing gathering algorithm in rings, assuming SSYNC, and that robots are able to exactly count how many of them occupy a particular location.
Castaneda et al.~\cite{ASMD2017} presented a weaker version of gathering, named edge-gathering. They provided a solution to edge-gathering in acyclic graphs, assuming that any number of robots may crash. On the one hand, their scheduling model is the most general (ASYNC); on the other hand, their robot model makes use of persistent memory (robots can remember some of their past actions, and communicate explicitly with other robots).

Overall, to our knowledge, the SUIR and SUIG problems were never addressed in the discrete setting. 
\subsection{Our Contribution}
In this paper, we initiate the research on SUIG and SUIR feasibility for robots on line-shaped graphs, considering the vanilla model (called OBLOT \cite{PGN2019}) where robots are oblivious (that is, they don't have access to persistent memory between activations), are \emph{not} able to distinguish multiple occupations of a given location, and can be completely disoriented (no common direction). 
More precisely, we focus on both of finite/infinite lines, and study conditions that preclude or enable SUIG and SUIR solvability. 
As in the continuous model, we first prove that SUIG and SUIR are impossible to solve in the SSYNC model, so we concentrate on the FSYNC model.
It turns out that, in FSYNC, SUIR is solvable if and only if the initial distance between the two robots is even, and that SUIG is solvable if only if the initial configuration is not edge-symmetric.
Our positive results are constructive, as we provide an algorithm for each case and prove it correct.
As expected, the key enabling algorithmic constructions we use for our protocols are fundamentally different from those used in continuous spaces~\cite{QAS2020,QAS2021}, as robots can no longer use fractional moves to solve the problem, and can be of independent interest to build further solutions in other topologies.

The rest of the paper is organized as follows. 
Section~\ref{sec:model} presents our model assumptions, Section~\ref{sec:suir} is dedicated to SUIR, and Section~\ref{sec:suig} is dedicated to SUIG. We provide concluding remarks in Section~\ref{sec:conclusion}. Due to space constraints, additional proofs and results are presented in the appendix.  
\section{Model}
\label{sec:model}
The line consists of an infinite or finite number of nodes $u_0, u_1, u_2, \dots$, such that a node $u_i$ is connected to both $u_{(i-1)}$ and $u_{(i+1)}$ (if they exist). Note that in the case where the line is finite of size $n$, two nodes of the line, $u_0$ and $u_{n-1}$ are only connected to $u_1$ and $u_{n-2}$, respectively.  

Let $R=\{r_1, r_2, \dots, r_k\}$ be the set of $k\geq 2$ autonomous robots.
Robots are assumed to be anonymous (i.e., they are indistinguishable), uniform (i.e., they all execute the same program, and use no localized parameter such as a particular orientation), oblivious (i.e., they cannot remember their past actions), and disoriented (i.e., they cannot distinguish left and right).
We assume that robots do not know the number of robots $k$.
In addition, they are unable to communicate directly, however, they have the ability to sense the environment including the positions of all other robots, i.e., they have infinite view. 
Based on the snapshot resulting of the sensing, they decide whether to move or to stay idle. 
Each robot $r$ executes cycles infinitely many times, \emph{(i)} first, $r$ takes a snapshot of the environment to see the positions of the other robots (LOOK phase), \emph{(ii)} according to the snapshot, $r$ decides whether it should move and where (COMPUTE phase), and \emph{(iii)} if $r$ decides to move, it moves to one of its neighbor nodes depending on the choice made in COMPUTE phase (MOVE phase). 
We call such cycles LCM (LOOK-COMPUTE-MOVE) cycles.
We consider the \emph{FSYNC} model in which at each time instant $t$, called round, each robot $r$ executes an LCM cycle synchronously with all the other robots, and the \emph{SSYNC} model where a non-empty subset of robots chosen by an adversarial scheduler executes an LCM cycle synchronously, at each $t$.

A node is considered \emph{occupied} if it contains at least one robot; otherwise, it is \emph{empty}. If a node $u$ contains more than one robot, it is said to have a \textit{tower} or \textit{multiplicity}. The ability to detect towers is called \textit{multiplicity detection}, which can be either \textit{global} (any robot can sense a tower on any node) or \textit{local} (a robot can only sense a tower if it is part of it). If robots can determine the number of robots in a sensed tower, they are said to have \textit{strong} multiplicity detection. In this work, we assume that robots do not have multiplicity detection and cannot distinguish between nodes with one robot and those with multiple robots.

As robots move and occupy nodes, their positions form the system's \emph{configuration} $C_t=(d(u_0), d(u_1), \dots)$ at time $t$. Here, $d(u_i) = 0$ if node $u_i$ is empty and $d(u_i) = 1$ if it is occupied.
  
Given two nodes $u_i$ and $u_j$, a segment $[u_i,u_j]$ represents the set of nodes between $u_i$ and $u_j$, inclusive. No assumptions are made about the state of the nodes in $[u_i,u_j]$. Any node $u \in [u_i, u_j]$ can be either empty or occupied. The number of occupied nodes in $[u_i,u_j]$ is represented by $|[u_i,u_j]|$. In Fig. \ref{fig:model}, each node $u_i$, where $i \in \{2, 3, \dots, 9\}$, is part of the segment $[u_2,u_9]$. Note that $|[u_2,u_9]|=3$ since nodes $u_2$, $u_5$, and $u_9$ are occupied.

For a given configuration $C$, let $u_i$ be an occupied node. Node $u_j$ is considered an \emph{occupied neighboring node} of $u_i$ in $C$ if it is occupied and if $|[u_i,u_j]|=2$. A \emph{border node} in $C$ is an occupied node with only one occupied neighboring node. A robot on a border node is referred to as a \emph{border robot}. In Fig. \ref{fig:model}, when $k=4$, nodes $u_0$ and $u_9$ are border nodes.

\begin{figure}[t]
    \begin{center}
        \begin{tikzpicture}
            \tikzstyle{robot}=[circle,fill=black!100,inner sep=0.12cm]
            \draw[-] (-1,0) -- (9,0);
            \def \n {9}
            \foreach \s in {0,...,\n}
            {
             \node[draw, circle, fill=white, inner sep=0.15cm] at (\s-0.5,0) {};
             \node[opacity=0, text opacity=1] at (\s-0.5,0.45) {$u_{\s}$};
            }
            \node[opacity=0, text opacity=1] at (5.2,1) {$[u_2,u_9]$};
\draw[<->] (1.3,0.7) -- (8.7,0.7);
            \node[robot] at (-0.5,0) {};
            \node[robot] at (1.5,0) {};
            \node[robot] at (4.5,0) {};
            \node[robot] at (8.5,0) {};
        \end{tikzpicture}
        \caption{Instance of a configuration in which the segment $[u_2,u_9]$ is highlighted.}\label{fig:model}
    \end{center}
\end{figure}
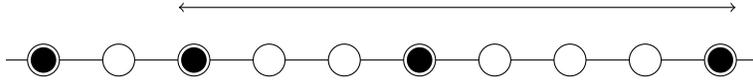

The \emph{distance} between two nodes $u_i$ and $u_j$ is the number of edges between them. 
The distance between two robots $r_i$ and $r_j$ is the distance between the two nodes occupied by $r_i$ and $r_j$, respectively.
We denote the distance between $u_i$ and $u_j$ (resp. $r_i$ and $r_j$) $dist(u_i,u_j)$ (resp. $dist(r_i,r_j)$).
Two robots or two nodes are \emph{neighbors} (or adjacent) if the distance between them is one.
A sequence of consecutive occupied nodes is a \emph{block}. 
Similarly, a sequence of consecutive empty nodes is a \emph{hole}.

An \emph{algorithm} $A$ is a function mapping the snapshot (obtained during the LOOK phase) to a neighbor node destination to move to (during the MOVE phase).  
An \emph{execution} ${\cal E}=(C_0, C_1,\dots)$ of $A$ is a sequence of configurations, where $C_0$ is an initial configuration, and every configuration $C_{t+1}$ is obtained from $C_{t}$ by applying $A$. 

Let $r$ be a robot located on node $u_i$ at time $t$ and let $S^+(t) = d(u_i), d(u_{i+1}),$ $\dots,$ $d(u_{i+m})$ and $S^-(t) = d(u_i), d(u_{i-1}), \dots d(u_{i-m'})$ be two sequences such that $m, m' \in \mathbb{N}$. 
Note that $m = n-i-1$ and $m' = i$ in the case where the line is finite of size $n$, $m=m'=\infty$ in the case where the line is infinite.
The view of robot $r$ at time $t$, denoted ${View}_r(t)$, is defined as the pair $\{S^+(t),S^-(t)\}$ ordered in the lexicographic order. 

Let $C$ be a configuration at time $t$. Configuration $C$ is said to be \emph{symmetric} at time $t$ if there exist two robots $r$ and $r'$ such that $View_r(t) = View_{r'}(t)$. 
In this case, $r$ and $r'$ are said to be symmetric robots. Let $C$ be a symmetric configuration at time $t$ then, $C$ is said to be \emph{node-symmetric} if the distance between two symmetric robots is even (i.e., if the axis of symmetry intersects with the line on a node), otherwise, $C$ is said to be \emph{edge-symmetric}. 
Finally, a non-symmetric configuration is called a \emph{rigid} configuration. 

\medskip
\noindent\textbf{Problem definition.}
A robot is said to be \emph{crashed} at time $t$ if it is not activated at any time $t^\prime \geq t$.
That is, a crashed robot stops execution and remains at the same position indefinitely.
We assume that robots cannot identify a crashed robot in their snapshots (i.e., they are able to see the crashed robots but remain unaware of their crashed status).
A crash, if any, can occur at any round of the execution. Furthermore, if more than one crash occurs, all crashes occur at the same location. In our model, since robots do not have multiplicity detection capability, a location with a single crashed robot and with multiple crashed robots are indistinguishable, and are thus equivalent. In the sequel, for simplicity, we consider at most one crashed robot.

We consider the \emph{Stand Up Indulgent Gathering} (SUIG) problem defined in \cite{QAS2021}. 
An algorithm solves the SUIG problem if, for any initial configuration $C_0$ (that may contain multiplicities), and for any execution ${\cal E}=(C_0, C_1,\dots)$, there exists a round $t$ such that all robots (including the crashed robot, if any) gather at a single node, not known beforehand, for all $t^\prime\geq t$.
The special case with $k=2$ is called the \emph{Stand Up Indulgent Rendezvous} (SUIR) problem.

\section{Stand Up Indulgent Rendezvous}\label{sec:suir}

We address in this section the case in which $k=2$, that is, the SUIR problem. 
We show that SSYNC solutions do not exist when one seeks a deterministic solution (Corollary~\ref{cor:impossible-SSYNC}), and that even in FSYNC, not all initial configurations admit a solution (Theorem~\ref{theo:SUIR-line}). 
By contrast, all other initial configurations admit a deterministic SUIR solution (Theorem~\ref{theo:SUIR-line-possible}).

\begin{theorem}\label{theo:SUIR-line}
Starting from a configuration where the two robots are at odd distance from each other on a line-shaped graph, the SUIR problem is unsolvable in FSYNC by deterministic oblivious robots without additional hypotheses.  
\end{theorem}

\begin{proof}
Let us first observe that in any configuration, both robots must move.
Indeed, if no robot moves, then no robot will ever move, and SUIR is never achieved. 
If one robot only moves, then the adversary can crash this robot, and the two robots never move, hence SUIR is never achieved.

In any configuration, two robots can either:
\emph{(i)} move both in the same direction,
\emph{(ii)} move both toward each other, or
\emph{(iii)} move both in the opposite direction.
Assuming an FSYNC scheduling and no crash by any robot, in the case of \emph{(i)}, the distance between the two robots does not change, in the case of \emph{(ii)}, the distance decreases by two, and in the case of \emph{(iii)}, it increases by two. 
Since the distance between the two robots is initially odd, then any FSYNC execution of a protocol step keeps the distance between robots odd. 
As a result, the distance never equals zero, and the robots never gather. 
\qed
\end{proof}

\begin{corollary}\label{cor:impossible-SSYNC}
The SUIR problem is unsolvable on a line in SSYNC without additional hypotheses.
\end{corollary}

\begin{proof}
For the purpose of contradiction, suppose that there exists a SUIR algorithm $A$ in SSYNC.
Consider an SSYNC schedule starting from a configuration where exactly one robot is activated in each round. 
Since any robot advances by exactly one edge per round, and that $A$ solves SUIR, every such execution reaches a configuration where the two robots are at distance $2i+1$, where $i$ is an integer, that is, a configuration where robots are at odd distance from one another.
From this configuration onward, the schedule becomes synchronous (as a synchronous schedule is still allowed in SSYNC). 
By Theorem~\ref{theo:SUIR-line}, rendezvous is not achieved, a contradiction.\qed
\end{proof}

By Theorem \ref{theo:SUIR-line}, we investigate the case of initial configurations where the distance between the two robots is even. 
It turns out that, in this case, the SUIR problem can be solved. 

\begin{theorem}\label{theo:SUIR-line-possible}
Starting from a configuration where the two robots are at even distance from each other, the SUIR problem is solvable in FSYNC by deterministic oblivious robots without additional hypotheses.  
\end{theorem}

\begin{proof}
Our proof is constructive. 
Consider the simple algorithm ``go to the other robot position.''

When no robot crashes, at each FSYNC round, the distance between the two robots decreases by two. 
Since it is initially even, it eventually reaches zero, and the robots stop moving (the other robot is on the same location), hence rendezvous is achieved. 

If one robot crashes, in the following FSYNC rounds, the distance between the two robots decreases by one, until it reaches zero. 
Then, the correct robot stops moving (the crashed robot is on the same location), hence rendezvous is achieved. 
\qed
\end{proof}

Observe that both Theorems~\ref{theo:SUIR-line} and \ref{theo:SUIR-line-possible} are valid regardless of the finiteness of the line network. 

\section{Stand Up Indulgent Gathering}\label{sec:suig}

We address the SUIG problem ($k\geq 2$) on line-shaped networks in the following. 
Section~\ref{sec:suig-impossible} first derives some impossibility results, and then Section~\ref{sec:suig-algorithm} presents our algorithm; finally, the proof of our algorithm appears in Section~\ref{sec:suig-proof}.

\subsection{Impossibility Results}\label{sec:suig-impossible}

\begin{theorem}[\cite{REA2008}]\label{thm:FSYNC-edge-impossible}
The gathering problem is unsolvable in FSYNC on line networks starting from an edge-symmetric configuration even with strong multiplicity detection.
\end{theorem}

\begin{corollary}
\label{cor:edge-symmetric}
The SUIG problem is unsolvable in FSYNC on line networks starting from an edge-symmetric configuration even with strong multiplicity detection.
\end{corollary}

\begin{proof}
Consider a FSYNC execution without crashes, and apply Theorem~\ref{thm:FSYNC-edge-impossible}.\qed
\end{proof}

As a result of Corollary~\ref{cor:edge-symmetric}, we suppose in the remaining of the section that initial configurations are \emph{not} edge-symmetric. 

\begin{lemma}
\label{lem:no-SSYNC-line}
Even starting from a configuration that is not edge-symmetric, the SUIG problem is unsolvable in SSYNC without additional hypotheses.
\end{lemma}

\begin{proof}
The proof is by induction on the number $X$ of occupied nodes. 
Suppose for the purpose of contradiction that there exists such an algorithm $A$.

If $X=2$, and if the distance between the two occupied nodes is $1$, consider an FSYNC schedule. 
All robots have the same view, so either all robots stay on their locations (and the configuration remains the same), or all robots go to the other location (and the configuration remains with $X=2$ and a distance of $1$ between the two locations). 
As this repeats forever, in both cases, the SUIG is not achieved, a contradiction.
If $X=2$, and if the distance between the two occupied nodes is at least $2$, consider a schedule that either $\emph{(i)}$ executes all robots on the first location, or
\emph{(ii)} executes all robots at the second location, at every round. 
So, the system behaves (as robots do not use additional hypotheses such as multiplicity detection) as two robots, initially on distinct locations separated by distance at least $2$. 
As a result, the distance between the two occupied nodes eventually becomes odd. 
Then, consider an FSYNC schedule, and by Theorem~\ref{theo:SUIR-line}, $A$ cannot be a solution, a contradiction.

Suppose now that for some integer $X$, the lemma holds. 
Let us show that it also holds for $X+1$. 
Consider an execution starting from a configuration with $X+1$ occupied nodes.
Since $A$ is a SUIG solution, any execution of $A$ eventually creates at least one multiplicity point by having a robot $r_1$ on one occupied node moved to an adjacent occupied node. 
Consider the configuration $C_b$ that is immediately before the creation of the multiplicity point. 
Then, in $C_b$, \emph{all} robots at the location of $r_1$ make a move (this is possible in SSYNC), so the resulting configuration has $X$ occupied nodes. 
By the induction hypothesis, algorithm $A$ cannot solve SUIG from this point, a contradiction.

So, for all possible initial configurations, algorithm $A$ cannot solve SUIG, a contradiction.\qed
\end{proof}

As a result of Lemma~\ref{lem:no-SSYNC-line}, we suppose in the sequel that the scheduling is FSYNC.

\subsection{Algorithm $\mathcal{A}_{L}$}\label{sec:suig-algorithm}

In the following, we propose an algorithm $\mathcal{A}_{L}$ in FSYNC such that the initial configuration is not edge-symmetric.

Before describing our strategy, we first provide some definitions that will be used throughout this section.  
In given configuration $C$, let $d_{C}$ be the largest even distance between any pair of occupied nodes and let $U_{C}$ be the set of occupied nodes at distance $d_{C}$ from another occupied node. 
If $C$ is node-symmetric, $U_C$ consists only of the two border nodes ($|U_C|=2$). 
Then, let $u$ and $u'$ be nodes in $U_C$.
By contrast, if $C$ is rigid, then $|U_C| \geq 2$ (refer to Lemma \ref{lem:even-exists}). 
Since each robot has a unique view in a rigid configuration, let $u \in U_C$ be the node that hosts the robots with the largest view among those on a node of $U_{C}$ and let $u'$ be the occupied node at distance $d_C$ from $u$. 
Observe that if there are two candidate nodes for $u'$, using the view of the robots again, we can uniquely elect one of the two which are at distance $d_C$ from $u$ (by taking the one with the largest view). 
That is, $u$ and $u'$ can be identified uniquely in $C$. 
We refer to $[u,u']$ by the target segment in $C$, and denote the number of occupied nodes in $[u,u']$ by $|[u,u']|$.
Finally, the set of occupied nodes which are not in $[u,u']$ is denoted by $O_{[u,u']}$.
Refer to Fig. \ref{fig:target-seg} for an example.  

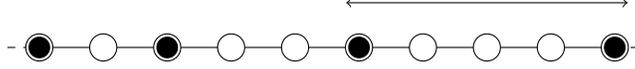
\begin{figure}[t]
    \begin{center}
        \begin{tikzpicture}[scale=0.85, transform shape]
            \tikzstyle{robot}=[circle,fill=black!100,inner sep=0.12cm]
            \draw[-] (-0.5,0) -- (8.5,0);
            \draw[dashed] (-1,0) -- (-0.5,0);
            \draw[dashed] (8.5,0) -- (9,0);
            \def \n {9}
            \foreach \s in {0,...,\n}
            {
             \node[draw, circle, fill=white, inner sep=0.15cm] at (\s-0.5,0) {};
            }
            \node[opacity=0, text opacity=1] at (4.5,0.4) {$u$};
            \node[opacity=0, text opacity=1] at (8.5,0.45) {$u'$};
            \node[opacity=0, text opacity=1] at (-0.5,0.4) {$v$};
            \node[opacity=0, text opacity=1] at (1.5,0.45) {$v'$};
            \node[opacity=0, text opacity=1] at (6.5,0.9) {target segment};
            \draw[<->] (4.3,0.7) -- (8.7,0.7);
            \node[robot] at (-0.5,0) {};
            \node[robot] at (1.5,0) {};
            \node[robot] at (4.5,0) {};
            \node[robot] at (8.5,0) {};
        \end{tikzpicture}
        \caption{Instance of a configuration in which $[u,u']$ is the target segment. Nodes $v$ and $v'$ are both in $O_{[u,u']}$. }\label{fig:target-seg}
    \end{center}
\end{figure}

We first observe that in any configuration $C$ in which there are at least three occupied nodes, there exist at least two occupied nodes at an even distance. 

\begin{lemma}\label{lem:even-exists}
In any configuration $C$ where there are at least three occupied nodes, there exists at least one pair of robots at an even distance from each other. 
\end{lemma}

\begin{proof}
Assume by contradiction that the lemma does not hold and let $u_1$, $u_2$, and $u_3$ be three distinct occupied nodes such that $u_2$ is located between $u_1$ and $u_3$. 
Assume w.l.o.g. that $d_1=dist(u_1,u_2)$ and $d_2=dist(u_2,u_3)$. 
By the assumption, both $d_1$ and $d_2$ are odd. 
That is, there exist $q, q' \in \mathbb{N}$, $d_1= 2q+1$ and $d_2= 2q'+1$. 
Hence, $d_1+d_2$ is even since $d_1+d_2= 2(q+q'+1)$.
A contradiction. \qed
\end{proof}

We propose, in the following, an algorithm named $\mathcal{A}_{L}$ that solves the SUIG problem on line-shaped networks. 
The main idea of the proposed strategy is to squeeze the robots by reducing the distance between the two border robots such that they eventually meet on a single node. 
To guarantee this meeting, robots aim to reach a configuration in which the border robots are at an even distance. 
Note that an edge-symmetric configuration can be reached during this process.
However, in this case, we guarantee that not only one robot has crashed but also, eventually, when there are only two occupied adjacent nodes, the crashed robot is alone on its node ensuring the gathering. 
Let $C$ be the current configuration.
In the following, we describe robots' behavior depending on $C$: 

\begin{enumerate}
\item $C$ is edge-symmetric.
The border robots move toward an occupied node. 
\item Otherwise.
Let $[u,u']$ be the target segment in $C$, and let $O_{[u,u']}$ be the set of robots located on a node, not part of $[u,u']$. 
Robots behave differently depending on the size of $O_{[u,u']}$:
    \begin{enumerate}
    \item $|O_{[u,u']}| = 0$ ($u$ and $u'$ host the border robots). 
    In this case, the border robots are the ones to move. 
    Their destination is their adjacent node toward an occupied node (refer to Fig. \ref{fig:0} for an example). 
        
    \item $|O_{[u,u']}| = 1$. 
    Let $u_v$ be the unique occupied node in $O_{[u,u']}$. 
    Assume w.l.o.g. that $u_v$ is closer to $u$ than to $u'$. 
    We address only the case in which $dist(u,u_v)$ is odd. 
    (Note that if $dist(u,u_v)$ is even, the border nodes are at an even distance and $|O_{[u,u']}| = 0$.) 
    \begin{itemize}
    \item If the number of occupied nodes is three and $u_v$ and $u$ are two adjacent nodes, robots on both $u$ and $u'$ are the ones to move. 
    Their destination is their adjacent node toward $u_v$ (refer to Fig. \ref{fig:1-bis}). 
    \item Otherwise, all robots are ordered to move. 
    More precisely, robots on a node of $[u,u']$ move to an adjacent node toward $u_v$ while the robots on $u_v$ move to an adjacent node toward $u$ (refer to Fig. \ref{fig:1}).
    \end{itemize}
    \item $|O_{[u,u']}| >1$. 
    In this case, the robots that are in $O_{[u,u']}$ move toward a node of $[u,u']$ (refer to Fig. \ref{fig:+1}). 
\end{enumerate}
\end{enumerate}

\begin{figure}[t]
    \begin{tabular}{c}
    \begin{minipage}[t]{\hsize}
    \begin{center}
        \begin{tikzpicture}[scale=0.85, transform shape]
            \tikzstyle{robot}=[circle,fill=black!100,inner sep=0.12cm]
            \draw[-] (-0.5,0) -- (7.5,0);
            \draw[dashed] (-1,0) -- (-0.5,0);
            \draw[dashed] (7.5,0) -- (8,0);
            \def \n {8}
            \foreach \s in {0,...,\n}
            {
             \node[draw, circle, fill=white, inner sep=0.15cm] at (\s-0.5,0) {};
            }
            \node[opacity=0, text opacity=1] at (-0.5,0.4) {$u$};
            \node[opacity=0, text opacity=1] at (7.5,0.45) {$u'$};
            \node[opacity=0, text opacity=1] at (3.5,0.9) {target segment};
            \node[draw, circle, inner sep=0.16cm, opacity=0] (m) at (0.5,0.0) {};
            \node[draw, circle, inner sep=0.16cm, opacity=0] (n) at (6.5,0.0) {};
            \draw[<->] (-0.75,0.7) -- (7.9,0.7);
            \node[robot] (a) at (-0.5,0) {};
            \node[robot] at (1.5,0) {};
            \node[robot] at (4.5,0) {};
            \node[robot] (b) at (7.5,0) {};
            \path (a) edge[bend left=50,->, red] node [left] {} (m);
            \path (b) edge[bend left=50,->, red] node [left] {} (n);
        \end{tikzpicture}
        \caption{Instance of a configuration in which $|O_{[u,u']}| = 0$. Robots at the border move toward an occupied node as shown by the red arrows.}\label{fig:0}
    \end{center}
\end{minipage}\\

\begin{minipage}[t]{\hsize}
    \begin{center}
        \begin{tikzpicture}[scale=0.85, transform shape]
            \tikzstyle{robot}=[circle,fill=black!100,inner sep=0.12cm]
         
            \draw[-] (-0.5,0) -- (8.5,0);
            \draw[dashed] (-1,0) -- (-0.5,0);
            \draw[dashed] (8.5,0) -- (9,0);
            \def \n {9}
            \foreach \s in {0,...,\n}
            {
             \node[draw, circle, fill=white, inner sep=0.15cm] at (\s-0.5,0) {};
            }
           \node[opacity=0, text opacity=1] at (0.5,0.4) {$u$};
            \node[opacity=0, text opacity=1] at (8.5,0.45) {$u'$};
            \node[opacity=0, text opacity=1] at (3.5,0.9) {target segment};
            \draw[<->] (0.3,0.7) -- (8.7,0.7);
            \node[robot] (a) at (-0.5,0) {};
            \node[robot] (a) at (0.5,0) {};
            \node[opacity=0, text opacity=1] at (-0.5,0.45) {$u_v$};
\node[robot] (d) at (8.5,0) {};
            \node[draw, circle, inner sep=0.16cm, opacity=0] (m) at (-0.5,0.0) {};
\node[draw, circle, inner sep=0.16cm, opacity=0] (p) at (7.5,0.0) {};
            
            \path (a) edge[bend left=50,->, red] node [left] {} (m);
\path (d) edge[bend left=50,->, red] node [left] {} (p);
        \end{tikzpicture}
        \caption{Instance of a configuration of the special case in which $|O_{[u,u']}| = 1$ with only three occupied nodes. Robots on $u$ and $u'$ move to their adjacent node toward the node in $O_{[u,u']}$, node $u_v$,  as shown by the red arrows. }\label{fig:1-bis}
    \end{center}
\end{minipage} \\

\begin{minipage}[t]{\hsize}
\begin{center}
        \begin{tikzpicture}[scale=0.85, transform shape]
            \tikzstyle{robot}=[circle,fill=black!100,inner sep=0.12cm]
         
            \draw[-] (-0.5,0) -- (8.5,0);
            \draw[dashed] (-1,0) -- (-0.5,0);
            \draw[dashed] (8.5,0) -- (9,0);
            \def \n {9}
            \foreach \s in {0,...,\n}
            {
             \node[draw, circle, fill=white, inner sep=0.15cm] at (\s-0.5,0) {};
            }
            \node[opacity=0, text opacity=1] at (2.5,0.4) {$u$};
            \node[opacity=0, text opacity=1] at (8.5,0.45) {$u'$};
            \node[opacity=0, text opacity=1] at (5.5,0.9) {target segment};
            \draw[<->] (2.3,0.7) -- (8.7,0.7);
            \node[robot] (a) at (-0.5,0) {};
            \node[opacity=0, text opacity=1] at (-0.5,0.45) {$u_v$};
            \node[robot] (b) at (2.5,0) {};
            \node[robot] (c) at (4.5,0) {};
            \node[robot] (d) at (8.5,0) {};
            \node[draw, circle, inner sep=0.16cm, opacity=0] (m) at (0.5,0.0) {};
            \node[draw, circle, inner sep=0.16cm, opacity=0] (n) at (1.5,0.0) {};
            \node[draw, circle, inner sep=0.16cm, opacity=0] (o) at (3.5,0.0) {};
            \node[draw, circle, inner sep=0.16cm, opacity=0] (p) at (7.5,0.0) {};
            
            \path (a) edge[bend left=50,->, red] node [left] {} (m);
            \path (b) edge[bend left=50,->, red] node [left] {} (n);
            \path (c) edge[bend left=50,->, red] node [left] {} (o);
            \path (d) edge[bend left=50,->, red] node [left] {} (p);
        \end{tikzpicture}
        \caption{Instance of a configuration in which $|O_{[u,u']}| = 1$. Robots in $[u,u']$ move to their adjacent node toward the robot in $O_{[u,u']}$ while the robots in $O_{[u,u']}$ move toward the target segment as shown by the red arrows. }\label{fig:1}
    \end{center}
\end{minipage}\\

\begin{minipage}[t]{\hsize}
    \begin{center}
        \begin{tikzpicture}[scale=0.85, transform shape]
            \tikzstyle{robot}=[circle,fill=black!100,inner sep=0.12cm]
          
            \draw[-] (-0.5,0) -- (8.5,0);
            \draw[dashed] (-1,0) -- (-0.5,0);
            \draw[dashed] (8.5,0) -- (9,0);
            \def \n {9}
            \foreach \s in {0,...,\n}
            {
             \node[draw, circle, fill=white, inner sep=0.15cm] at (\s-0.5,0) {};
            }
            \node[opacity=0, text opacity=1] at (4.5,0.4) {$u$};
            \node[opacity=0, text opacity=1] at (8.5,0.45) {$u'$};
            \node[opacity=0, text opacity=1] at (6.5,0.9) {target segment};
            \draw[<->] (4.3,0.7) -- (8.7,0.7);
            \node[robot] (a) at (-0.5,0) {};
            \node[robot] (b) at (1.5,0) {};
            \node[robot] at (4.5,0) {};
            \node[robot] at (8.5,0) {};
            \node[draw, circle, inner sep=0.16cm, opacity=0] (m) at (0.5,0.0) {};
            \node[draw, circle, inner sep=0.16cm, opacity=0] (n) at (2.5,0.0) {};
            \path (a) edge[bend left=50,->, red] node [left] {} (m);
            \path (b) edge[bend left=50,->, red] node [left] {} (n);
        \end{tikzpicture}
        \caption{Instance of a configuration in which $|O_{[u,u']}| > 1$. Robots in $O_{[u,u']}$ move to their adjacent node toward the target segment as shown by the red arrows. }\label{fig:+1}
    \end{center}
\end{minipage}
\end{tabular}
\end{figure}

\subsection{Proof of the Correctness}\label{sec:suig-proof}

We prove in the following the correctness of $\mathcal{A}_L$. 

\begin{lemma}\label{lem:same-side}
Let $C$ be a non-edge-symmetric configuration, and let $[u,u']$ be the target segment. If $|O_{[u,u']}|>1$, then all nodes in $O_{[u,u']}$ are located on the same side of $[u,u']$. 
\end{lemma}

\begin{proof}
Assume by contradiction that the lemma does not hold and assume that there exists a pair of nodes $u_1, u_2 \in O_{[u,u']}$ such that $u_1$ and $u_2$ are on different sides of $[u,u']$. 
Assume w.l.o.g. that $u_1$ is the closest to $u$ and $u_2$ is the closest to $u'$. 
Let $dist(u_1,u) = d_1$ and $dist(u_2,u')=d_2$. 
If $d_1$ and $d_2$ are both even or both odd, $dist(u_1,u_2)$ is even. 
Since $dist(u_1,u_2) > dist(u,u')$, $[u,u']$ is not the target segment, a contradiction.
Otherwise, assume w.l.o.g. that $d_1$ is even and $d_2$ is odd, then $dist(u_1,u')$ is even. Since $dist(u_1,u') > dist(u,u')$, $[u,u']$ is not the target segment, a contradiction.  
\qed
\end{proof}

\begin{theoremEnd}[sss]{lemma}\label{theo:no-crash-ok}
Starting from a non-edge-symmetric configuration $C$, if no robot crashes, all robots gather without multiplicity detection in $O(\mathcal{D})$ by executing $\mathcal{A}_L$, where $\mathcal{D}$ denotes the distance between the two borders in $C$.
\end{theoremEnd}

\begin{proofEnd}
Assume by contradiction that the theorem does not hold. 
Let $C_t$ be the current configuration. The following two cases are possible:  

\begin{enumerate}
\item \label{case:border-even} $\mathcal{D}$ is even.
Let $u_1, u_2, \dots, u_m$ be the sequence of nodes between the two border robots such that $u_1$ and $u_m$ host a border robot. 
Note that the target segment is, in this case, $[u_1,u_m]$ and $|O_{[u_1,u_m]}|=0$.
By $A_L$, robots on $u_1$ and $u_m$ are the ones to move toward respectively $u_2$ and $u_{m-1}$. 
By moving, the distance between the border robots decreases by two and hence remains even. 
By $A_L$, robots on $u_2$ and $u_{m-1}$ are now the new borders and are the ones to move to, respectively, $u_3$ and $u_{m-2}$. 
As previously, the distance between the border robots decreases by two. 
By repeating the same process, after $\lfloor\frac{\mathcal{D}-1}{2}\rfloor$ rounds, the border robots become at distance two from each other on respectively nodes $u_{\frac{\mathcal{D}}{2}}$ and $u_{m-(\frac{\mathcal{D}}{2}+1)}$. 
After one additional round, all robots meet on $u_{(\frac{\mathcal{D}}{2})+1}$. 
Hence, the robots gather after $\frac{\mathcal{D}}{2}$ rounds. 
A contradiction.  
    
\item $\mathcal{D}$ is odd.
By Lemma \ref{lem:even-exists}, we know that there is at least one pair of robots which are at an even distance from each other. 
Moreover, $C_t$ is not symmetric as we retrieve case \ref{case:border-even} otherwise. 
Let $u$ and $u'$ be the two occupied nodes with respect to $\mathcal{A}_L$ such that $[u,u']$ is the target segment. 
Assume w.l.o.g. that $v$, the border robot in $O_{[u,u']}$, is closer to $u$.
Let $\mathcal{D}' = dist(v,u)$. 
Two cases are possible: 
    \begin{enumerate}
        \item $|O_{[u,u']}|>1$. 
        Then, all nodes in $O_{[u,u']}$ are on the same side of $[u,u']$ by Lemma \ref{lem:same-side}. 
        By the assumption, they are closer to $u$.
        By $A_L$, robots on the nodes in $O_{[u,u']}$ in $C_t$ move toward the nodes of $[u,u']$. 
        That is, after one round, $v$ becomes at an even distance from $u'$, and we retrieve case \ref{case:border-even}. Hence, we can deduce that the gathering is achieved after $\frac{\mathcal{D}-1}{2}+1$ rounds. 
        A contradiction.  
        
        \item $|O_{[u,u']}|=1$. 
        By the assumption and w.l.o.g., $dist(v,u) < dist(v, u')$. 
        Let $u_{v_0}, u_{v_1}, u_{v_2}, \dots u_{v_x}, \dots, u_{v_x'}$ be the sequence of nodes from $u_{v_0}$ toward $u$ such that $u_{v_0}=v$, $u_{v_x}=u$ and $u_{v_x'} = u'$ (refer to Fig. \ref{fig:1exp} for an example). 
        Let $\mathcal{D}'=dist(u_{v_0},u_{v_x})$ (in Fig. \ref{fig:1exp} - configuration $C_{t_0}$, $u_{v_x} = u_{v_5}$ and $\mathcal{D}'=5$). 
        By $\mathcal{A}_L$, two cases are possible:
        \begin{enumerate}
            \item The number of occupied nodes is equal to three and $u_{v_0}$ is adjacent to $u_{v_x}$ then, the robots on $[u,u']$ ($[u_{v_x},u_{v_x'}]$) move toward $u_{v_0}$. 
            As there is no crashed robot, after one round, the border robots become at an even distance, and we retrieve case \ref{case:border-even}. 
            Thus, the gathering is achieved in $\frac{\mathcal{D}-1}{2}+1$ rounds.
            
            \item Otherwise, all robots on a node in $[u,u']$ move toward $u_{v_0}$ while the robots on node $u_{v_0}$ move toward $u$. 
            Let $C_{t_1}$ be the configuration reached once the robots move. 
            Since there is no crashed robot by assumption, all the robots move and the distance between every robot on a node in $[u,u']$ and those on $u_v$ decreases by two. 
            Hence, $[u_{v_{x-1}}, u_{v_{x'-1}}]$ become the target segment in $C_{t_1}$ and $dist(u,u') = dist(u_{v_{x-1}}, u_{v_{x'-1}})$ in $C_{t1}$. 
            Moreover, $dist(u_{v_0}, u)$ in $C_t$ is equal to $dist(u_{v_1}, u_{v_{x-1}})-2$ in $C_{t_1}$. 
            By $A_L$, the robots on a node in $[u_{v_{x-1}}, u_{v_{x'-1}}]$ move toward $u_{v_1}$ while the robots on $u_{v_1}$ move toward $u_{v_{x-1}}$. 
            In the configuration reached $C_{t_2}$, $[u_{v_{x-2}}, u_{v_{x'-2}}]$ is the target segment. 
            By repeating the same process, after ${\lfloor \frac{\mathcal{D}'}{2} \rfloor}$ rounds, in configuration $C_{t_{{\lfloor \frac{x}{2} \rfloor}}}$, the robots initially on $u$ in $C_t$ become located on node $u_{v_{\lceil \frac{x}{2} \rceil}}$, the ones on $u_{v_{0}}$ on $u_{v_{\lfloor \frac{x}{2} \rfloor}}$ and the ones on $u'$ on node $u_{v_{x' -\lfloor \frac{x}{2} \rfloor}}$ (refer to Fig. \ref{fig:1exp} for an example). 
            Again, as there is no crashed robot, $dist(u,u')$ in $C_t$ is equal to $dist(u_{v_{\lceil \frac{x}{2} \rceil}}, u_{v_{x' -\lfloor \frac{x}{2} \rfloor}})$ and $[u_{v_{\lceil \frac{x}{2} \rceil}}, u_{v_{x' -\lfloor \frac{x}{2} \rfloor}}]$ is the target segment in  $C_{t_{\lfloor\frac{x}{2} \rfloor}}$. 
            After one additional round, we retrieve case \ref{case:border-even}.
            We can deduce that the gathering is achieved after $\frac{\mathcal{D}-1}{2}$ rounds, $\frac{\mathcal{D'}-1}{2}$ rounds are needed for the borders to become at an even distance and then $\frac{\mathcal{D}-\mathcal{D}'}{2}$ rounds for the robots to gather. 
        \end{enumerate}
        
    \end{enumerate}
      
From the cases above, we can deduce that robots gather in $O(\mathcal{D})$ rounds.\qed
\end{enumerate}

\begin{figure}[t]\begin{center}
        \begin{tikzpicture}[scale=0.8, transform shape]
            \node[opacity=0, text opacity=1] at (-1.5,0) {$C_{t_0}$};
            \tikzstyle{robot}=[circle,fill=black!100,inner sep=0.12cm];
            
            \draw[dashed] (-1,0) -- (-0.5,0);
            \draw[dashed] (10.5,0) -- (11,0);
            \draw[-] (-0.5,0) -- (10.5,0);
            \def \n {11}
            \foreach \s in {0,...,\n}
            {
             \node[draw, circle, fill=white, inner sep=0.15cm] at (\s-0.5,0) {};
            }
          
            \node[opacity=0, text opacity=1] at (4.5,0.4) {$u$};
            \node[opacity=0, text opacity=1] at (10.5,0.45) {$u'$};
            \node[opacity=0, text opacity=1] at (7,0.9) {target segment};
            \draw[<->] (4.3,0.7) -- (10.7,0.7);
            \node[robot] (a) at (-0.5,0) {};
            
            \foreach \s in {0,...,\n}
            {
            \node[opacity=0, text opacity=1] at (\s-0.5,-0.45) {$u_{v_{\s}}$};
}
            
            \node[opacity=0, text opacity=1] at (-0.5,-0.45) {$u_{v_0}$};
            \node[robot] (b) at (4.5,0) {};
            \node[robot] (c) at (6.5,0) {};
            \node[robot] (d) at (10.5,0) {};
            \node[draw, circle, inner sep=0.16cm, opacity=0] (m) at (0.5,0.0) {};
            \node[draw, circle, inner sep=0.16cm, opacity=0] (n) at (3.5,0.0) {};
            \node[draw, circle, inner sep=0.16cm, opacity=0] (o) at (5.5,0.0) {};
            \node[draw, circle, inner sep=0.16cm, opacity=0] (p) at (9.5,0.0) {};
            
            \path (a) edge[bend left=50,->, red] node [left] {} (m);
            \path (b) edge[bend left=50,->, red] node [left] {} (n);
            \path (c) edge[bend left=50,->, red] node [left] {} (o);
            \path (d) edge[bend left=50,->, red] node [left] {} (p);
        \end{tikzpicture}
        
        \begin{tikzpicture}[scale=0.8, transform shape]
        \node[opacity=0, text opacity=1] at (-1.5,0) {$C_{t_1}$};
            \tikzstyle{robot}=[circle,fill=black!100,inner sep=0.12cm];
            \draw[dashed] (-1,0) -- (-0.5,0);
            \draw[dashed] (10.5,0) -- (11,0);
            \draw[-] (-0.5,0) -- (10.5,0);
            \def \n {11}
            \foreach \s in {0,...,\n}
            {
             \node[draw, circle, fill=white, inner sep=0.15cm] at (\s-0.5,0) {};
            }
\node[opacity=0, text opacity=1] at (6,0.9) {target segment};
            \draw[<->] (3.3,0.7) -- (9.7,0.7);
            \node[robot] (a) at (0.5,0) {};
            
            \foreach \s in {0,...,\n}
            {
            \node[opacity=0, text opacity=1] at (\s-0.5,-0.45) {$u_{v_{\s}}$};
}
            
\node[robot] (b) at (3.5,0) {};
            \node[robot] (c) at (5.5,0) {};
            \node[robot] (d) at (9.5,0) {};
            \node[draw, circle, inner sep=0.16cm, opacity=0] (m) at (1.5,0.0) {};
            \node[draw, circle, inner sep=0.16cm, opacity=0] (n) at (2.5,0.0) {};
            \node[draw, circle, inner sep=0.16cm, opacity=0] (o) at (4.5,0.0) {};
            \node[draw, circle, inner sep=0.16cm, opacity=0] (p) at (8.5,0.0) {};
            
            \path (a) edge[bend left=50,->, red] node [left] {} (m);
            \path (b) edge[bend left=50,->, red] node [left] {} (n);
            \path (c) edge[bend left=50,->, red] node [left] {} (o);
            \path (d) edge[bend left=50,->, red] node [left] {} (p);
        \end{tikzpicture}
        
        \begin{tikzpicture}[scale=0.8, transform shape]
            \node[opacity=0, text opacity=1] at (-1.5,0) {$C_{t_2}$};
            \tikzstyle{robot}=[circle,fill=black!100,inner sep=0.12cm]
            \draw[dashed] (-1,0) -- (-0.5,0);
            \draw[dashed] (10.5,0) -- (11,0);
            \draw[-] (-0.5,0) -- (10.5,0);
            \def \n {11}
            \foreach \s in {0,...,\n}
            {
             \node[draw, circle, fill=white, inner sep=0.15cm] at (\s-0.5,0) {};
            }
\node[opacity=0, text opacity=1] at (5,0.9) {target segment};
            \draw[<->] (2.3,0.7) -- (8.7,0.7);
            \node[robot] (a) at (1.5,0) {};
            
            \foreach \s in {0,...,\n}
            {
            \node[opacity=0, text opacity=1] at (\s-0.5,-0.45) {$u_{v_{\s}}$};
}
            
\node[robot] (b) at (2.5,0) {};
            \node[robot] (c) at (4.5,0) {};
            \node[robot] (d) at (8.5,0) {};
            \node[draw, circle, inner sep=0.16cm, opacity=0] (m) at (2.5,0.0) {};
            \node[draw, circle, inner sep=0.16cm, opacity=0] (n) at (1.5,0.0) {};
            \node[draw, circle, inner sep=0.16cm, opacity=0] (o) at (3.5,0.0) {};
            \node[draw, circle, inner sep=0.16cm, opacity=0] (p) at (7.5,0.0) {};
            
            \path (a) edge[bend left=50,->, red] node [left] {} (m);
            \path (b) edge[bend left=50,->, red] node [left] {} (n);
            \path (c) edge[bend left=50,->, red] node [left] {} (o);
            \path (d) edge[bend left=50,->, red] node [left] {} (p);
        \end{tikzpicture}
        
         \begin{tikzpicture}[scale=0.8, transform shape]
         \node[opacity=0, text opacity=1] at (-1.5,0) {$C_{t_3}$};
            \tikzstyle{robot}=[circle,fill=black!100,inner sep=0.12cm];
            \draw[dashed] (-1,0) -- (-0.5,0);
            \draw[dashed] (10.5,0) -- (11,0);
            \draw[-] (-0.5,0) -- (10.5,0);
            \def \n {11}
            \foreach \s in {0,...,\n}
            {
             \node[draw, circle, fill=white, inner sep=0.15cm] at (\s-0.5,0) {};
            }
\node[opacity=0, text opacity=1] at (4,0.9) {target segment};
            \draw[<->] (1.3,0.7) -- (7.7,0.7);
            \node[robot] (a) at (2.5,0) {};
            
            \foreach \s in {0,...,\n}
            {
            \node[opacity=0, text opacity=1] at (\s-0.5,-0.45) {$u_{v_{\s}}$};
}
            
\node[robot] (b) at (1.5,0) {};
            \node[robot] (c) at (3.5,0) {};
            \node[robot] (d) at (7.5,0) {};
            \node[draw, circle, inner sep=0.16cm, opacity=0] (m) at (2.5,0.0) {};
\node[draw, circle, inner sep=0.16cm, opacity=0] (p) at (6.5,0.0) {};
            
            \path (b) edge[bend left=50,->, red] node [left] {} (m);
\path (d) edge[bend left=50,->, red] node [left] {} (p);
        \end{tikzpicture}
        \caption{Instance of an execution starting from a configuration in which $|O_{[u,u']}| = 1$ and assuming no crashes.}\label{fig:1exp}
    \end{center}
\end{figure}
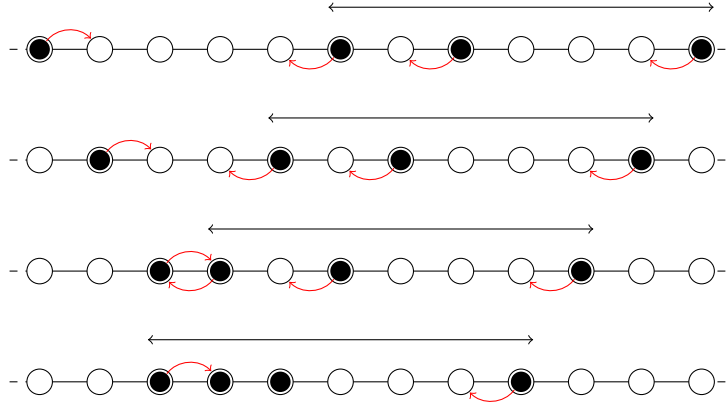

\end{proofEnd}

We focus in the following on the case in which a single robot crashes.

\begin{lemma}\label{lem:gathering}
Starting from a configuration $C$ where there are only two occupied nodes at distance $\mathcal{D}>1$, one hosting the crashed robot, gathering is achieved in $\mathcal{D}$ rounds, where $\mathcal{D}$ denotes the distance between the two borders in $C$.
\end{lemma}

\begin{proof}
First observe that if the crashed robot is not collocated with a non-crashed robot, after one round the robots on the other border move towards the crashed robot location, and gathering is achieved after $\mathcal{D}$ rounds.

Now assume that the crashed robot is collocated with at least one non-crashed robot. 
If $\mathcal{D}=2$, then after one round, all non-crashed robots are located at the same node, adjacent to the crashed robot location, and after one more round, gathering is achieved.
If $\mathcal{D}>3$ is even, then after one round the crashed robot is alone, and the non-crashed robots form two multiplicity points and are the extremities of the target segment. So, after one more round, they both move towards the crashed robot location, and we reach a configuration with two occupied nodes, and the distance between them has decreased by two. By induction and the previous case, gathering is eventually achieved.
If $\mathcal{D}=3$, then similarly, one can show that in three rounds, gathering is achieved.
If $\mathcal{D}>3$ is odd, then similarly, one can show that after three rounds we reach configuration with two occupied nodes, and their distance has decreased by 3 (so, the distance is now even and we can apply one of the previous even cases).
\qed
\end{proof}

\begin{theoremEnd}[sss]{lemma}\label{lem:adjacent}
Let $C$ be a configuration where $[u,u']$ is the target segment with $O_{[u,u']} = \{v\}$, and w.l.o.g. $u'$ is a border. Let $\mathcal{D}$ be the distance between the two border nodes $u'$ and $v$. If $u'$ hosts a crashed robot and $dist(v,u)=1$, then gathering is achieved in $O(\mathcal{D})$ rounds.
\end{theoremEnd}

\begin{proofEnd}
Thanks to Lemma~\ref{lem:gathering}, it is enough to prove that, in $O(\mathcal{D})$ rounds, a configuration $C'$ in which there are only two occupied nodes $v$ and $u'$ with $dist(v,u')>1$ is reached. 

By $\mathcal{A}_L$, if the number of occupied nodes in $C$ is more than three, the robots in $[u,u']$ move toward $v$ while the robots on $v$ move toward a node toward $[u,u']$. 
That is, if there is a multiplicity on $u'$, all non-crashed robots move to their adjacent node toward $u$. 
Thus, after one round, $u'$ hosts only a crashed robot. 
As $u$ and $v$ are neighbors, the robots on these nodes simply exchange their positions. 
However, note that $[u,u']$ remains the target segment and the distance between the robots on $u$ and nodes $v'' \in [u,u']$ with $v''\ne u$ and $v'' \ne u'$ decreases.
In the reached configuration, the same robots are ordered to move. 
Hence, eventually, $u$ becomes adjacent to two occupied nodes. 
After one round, $|[u,u']|$ decreases. 
As $u$ and $v$ remain occupied and adjacent to each other, the robots in $[u,u']$ continue to move toward $u$ to eventually join it. 
Thus, by repeating this process, $|[u,u']|$ decreases until $|[u,u']|=2$.  
As the initial distance between the robots is less than $\mathcal{D}$, after at most $\mathcal{D}$ rounds, a configuration in which there are three occupied nodes $u$, $u'$ and $v$, is reached. 
In this special case, by $\mathcal{A}_L$, after one round, $v$ and $u'$ are the only occupied nodes. 
Observe that since $[u,u']$ is a target segment, $dist(u,u') \geq 2$ and hence $dist(v,u') \geq 3$. 
After that, by $\mathcal{A}_L$, robots on $v$ are the only ones to move. 
Hence the lemma holds.\qed
\end{proofEnd}

\begin{lemma}\label{lem:Notadjacent}
Let $C$ be a configuration where $[u,u']$ is the target segment, $|O_{[u,u']}|=1$ and w.l.o.g. $dist(u,v) < dist(u',v)$ where $v \in O_{[u,u']}$.  
Let $\mathcal{D}$ be the distance between the two border nodes $u'$ and $v$ in $C$.
If $u'$ hosts a crashed robot, then gathering is achieved in $O(\mathcal{D})$ rounds.
\end{lemma}

\begin{proof}
Recall that, since $v \in O_{[u,u']}$, $dist(u',v)$ is odd in $C$. If the distance $m$ between $u$ and $v$ is 1, the we apply Lemma~\ref{lem:adjacent}, otherwise, 
by $\mathcal{A}_L$, the robots in $v$ move to an empty node toward $u$, and all the other robots move towards $v$.
Thus, after one round, we reach a configuration $C_1$ where the robots on $v$ becomes at an even distance from $u'$, the crashed robot location.
Since the robots on $u$ are also ordered to move toward $v$, the distance between the robots at $v$ and $u$ decreases by two.

Again, as $u'$ hosts a crashed robot, after one more round, a configuration $C_2$ in which the borders are at an odd distance is reached again, and the distance between the robots at $u$ and $v$ is again decreased by two (or stay the same if they are adjacent in $C_2$ as they just swap their positions).

As the distance $m$ between $v$ and $u$ is odd in $C$ (otherwise, the border robots are at an even distance in $C$), we can repeat the same 2-round process ($\lceil m/4\rceil$ times) until we reach a configuration in which the robots at $u$ and $v$ are at distance 1 so, by Lemma \ref{lem:adjacent}, gathering is achieved. 
\qed
\end{proof}

\begin{theoremEnd}[sss]{lemma}\label{lem:border-even}
Starting from a configuration $C$ where the crashed robot is at a border, and the border robots are at an even distance $\mathcal{D}$, by executing $\mathcal{A}_L$, after $O(\mathcal{D}$) rounds, gathering is achieved. \end{theoremEnd}

\begin{proofEnd}
The proof is by induction on $\mathcal{D}$.
As the borders are at an even distance, by $\mathcal{A}_L$, these robots are the ones to move. 
However, as the crashed robot does not move, after one round, the distance between the border robots $\mathcal{D}_1$ becomes odd ($\mathcal{D}_1 \geq 3$). 
Observe that $\mathcal{D}_1= \mathcal{D}-1$. 

So if $\mathcal{D} = 2$, in the reached configuration $C''$, all the non-crashed robots are located at a node adjacent to the crashed robot location, and after one more round they all move towards the crashed robot location and the gathering is achieved.

If $\mathcal{D} > 2$, in the reached configuration $C''$, the following cases are possible:  
\begin{enumerate}
\item \label{case:edge-edge} The configuration $C''$ is edge-symmetric. 
By $\mathcal{A}_L$, the border robots are the ones to move and their destination is their adjacent node toward an occupied node.
Hence, at the next round, all robots located at one border move to their adjacent node toward an occupied node. 
In the configuration, the distance between the two border robots $\mathcal{D}_2$ is even and $\mathcal{D}_2=\mathcal{D}-2$.
    
\item \label{case:nosym}Otherwise. 
By Lemma \ref{lem:even-exists}, there exists a pair of occupied nodes that are at an even distance. 
As $C''$ is neither node-symmetric (otherwise, the borders are at even distance) nor edge-symmetric (otherwise, case~\ref{case:edge-edge} holds), two occupied nodes $u$ and $u'$ are uniquely identified such that $[u,u']$ is the target segment. 
Let $u_x$ be the node that hosts the crashed robot, two cases are possible: 
  \begin{enumerate}
      \item $u_x \in \{u,u'\}$. 
      Assume w.l.o.g. that $u_x = u'$. 
      Let us first consider the case in which $|O_{[u,u']}|=1$ and let $v \in O_{[u,u']}$. 
      If $dist(u,v)=1$, we are done by Lemma \ref{lem:adjacent}. 
      By contrast, if $dist(u,v)>1$ then, let $u, u_1, \dots u_m, v$ be the sequence of nodes from $u$ to $v$. 
      By Lemma \ref{lem:Notadjacent}, after one round, a configuration $C'''$ 
      
      Finally, if $|O_{[u,u']}|>1$, after one round, a configuration in which the borders are at an even distance $\mathcal{D}_2= \mathcal{D}-2$ is reached and we apply the induction hypothesis. 
      
      \item $u_x \in O_{[u,u']}$. 
      Assume w.l.o.g. that $dist(u,u_x)<dist(u',u_x)$ holds.
      Let $u_1$ be the closest occupied node to $u_x$, and $u_0$ be an adjacent node on the side of $u_1$. Consider the case when $|O_{[u,u']}|>1$ holds.
      \begin{itemize}
      \item If $u_x$ hosts a multiplicity, then after one round, the non-crashed robots move to $u_0$ and a configuration in which $[u',u_0]$ is the target segment is reached with $|O_{[u',u_0]}|=1$. 
      \item If $u_x$ hosts only a single robot (the crashed one), as only robots on $u_1$ are ordered to move toward $u$, after one round, a configuration in which $[u',u_1]$ is the target segment is reached with $|O_{[u',u_1]}|=1$.
      \end{itemize}
      Hence, in both scenarios, we retrieve the case in which $|O_{[u,u']}| = 1$. Let us now focus on the case in which $|O_{[u,u']}| = 1$. 
      By $\mathcal{A}_L$, all robots move toward $u_x$ if $u_0 \ne u_1$. 
      Hence, $u_0$ becomes occupied eventually.  
      After one round, the distance between two border robots $\mathcal{D}_2$ become even, $\mathcal{D}_2= \mathcal{D}-2$, and we can use the induction hypothesis.
  \end{enumerate}
 \end{enumerate}
\qed
 \end{proofEnd}

\begin{lemma}\label{theo2:crashed-ok}
Starting from a non-edge-symmetric configuration $C$ with one crashed robot, all robots executing $\mathcal{A}_L$ eventually gather without multiplicity detection in $O(\mathcal{D})$ rounds, where $\mathcal{D}$ denotes the distance between the two borders in $C$. 
\end{lemma}

\begin{proof}
First, let us consider the case where $\mathcal{D}$ is even.
\begin{enumerate}
\item If the crashed robot is at an equal distance from both borders, 
by $\mathcal{A}_L$, the border robots move toward each other. 
As they do, the distance between them remains even. 
Hence, the border robots remain the only ones to move. 
Eventually, all robots which are not co-located with the crashed robot become border robots and hence move. Thus, the gathering is achieved in $\frac{\mathcal{D}}{2}$ rounds.
\item \label{case:crashed-border} If the crashed robot is a border, by Lemmas \ref{lem:border-even}, we can deduce that the gathering is achieved in $O(\mathcal{D})$ rounds.
\item Otherwise, as the border robots move toward each other by $\mathcal{A}_L$, the crashed robot eventually becomes at the border. 
We hence retrieve case \ref{case:crashed-border}.  
\end{enumerate}

From the cases above, we can deduce that the gathering is achieved whenever a configuration in which the borders are at an even distance, is reached. 

Let us now focus on the case where $\mathcal{D}$ is odd. By $\mathcal{A}_L$, two occupied nodes $u$ and $u'$ at the largest even distance are uniquely identified to set the target segment $[u,u']$ (recall that $C$ is, in this case, rigid, and each robot has a unique view since the initial configuration cannot be edge-symmetric). 
The robots behave differently depending on the size of $O_{[u,u']}$, the set of occupied nodes outside the segment $[u,u']$. 
By Lemma \ref{lem:same-side}, all nodes in $O_{[u,u']}$ are on the same side. 
Assume w.l.o.g. that for all $u_i \in O_{[u,u']}$, $u_i$ is closer to $u$ than $u'$. 
Two cases are possible: 

\begin{enumerate}
    \item \label{case:theo1} $|O_{[u,u']}|>1$. 
    Let $u_f, u_{f'} \in O_{[u,u']}$ be the two farthest nodes from $u$ such that $dist(u,u_f)>dist(u,u_{f'})$. Note that $u_f$ is a border robot. Let $u_{f+1}$ be $u_f$'s adjacent node toward $u_f'$. 
    If $u_f$ does not host a crashed robot, then as the robots on $u_f$ move toward $u$ and those on $u'$ remain idle by $\mathcal{A}_L$, after one round, the border robots become at an even distance and we are done. 
    By contrast, if $u_f$ hosts a crashed robot, then either $u_f$ hosts other non-crashed robots, and hence after one round, we retrieve a configuration $C'$ in which $[u_{f+1},u']$ is the target segment and $|O_{[u_{f+1},u']}|=1$ or a configuration $C'$ in which $[u_{f'},u']$ is the target segment and $|O_{[u_{f'},u']}|=1$ as robots on $u_{f'}$ also move toward $u$ by $\mathcal{A}_L$. In both cases, we retrieve the following case. 

    \item \label{case:theo2} $|O_{[u,u']}| = 1$. 
    Let $u_f$ be in $O_{[u,u']}$ and assume w.l.o.g. that $dist(u,u_f)<dist(u',u_{f})$ (observe that $u_f$ is a border node).
    If $u_f$ hosts the crashed robot, then, after one round, the border robots are at an even distance as robots on $u'$ move toward $u_f$ by $\mathcal{A}_L$.  
    Similarly, if $u'$ hosts the crashed robot, then by Lemmas \ref{lem:adjacent} and \ref{lem:Notadjacent} after $O(\mathcal{D})$ rounds, a configuration in which the border robots are at an even distance is reached. 
    If neither $u_f$ nor $u'$ hosts the crashed robot, then when the border robots move by $\mathcal{A}_L$ (other robots also move, but we focus for now on the border robots), either the distance between the two borders becomes even after one round (in the case where $u_f$ is adjacent to $u$ as the robots simply exchange their respective positions) or the distance between the border robots remains odd but decreases by two. Observe that in the later case, the border robots keep moving toward each other by $\mathcal{A}_L$ until one of them becomes a neighbor to a crash robot. 
    
    Let $C'$ be the configuration reached once a border robot becomes adjacent to a crashed robot. 
    Let $u_b$ be the border node that is adjacent to crashed robot, and let $u_{\overline{b}}$ be the other border node. 
    We refer to the node that hosts the crashed robot by $u_c$. 
    Since $dist(u_b,u_c)=1$ and $dist(u_b,u_{\overline{b}})$ is odd, $dist(u_{\overline{b}}, u_c)$ is even. Hence, $|O_{[u_{\overline{b}}, u_c]}|=1$. Two cases are possible: 
    
    \begin{itemize}
    \item If $C'$ hosts only three occupied nodes, then, after one round, the distance between the border robots becomes even, and we are done (recall that robots on $u_{\overline{b}}$ and $u_c$ move toward $u_b$ by $\mathcal{A}_L$). 
    \item If there are more than 3 occupied nodes and $u_c$ hosts also non-crashed robots in $C'$, then after one round, the distance between the border robots becomes even as the robots on $u_{\overline{b}}$ move toward $u_{c}$, those on $u_{b}$ move to $u_{c}$ and the non-crashed robots on $u_{c}$ move toward $u_{b}$ by $\mathcal{A}_L$. Hence, we are done.
    \item Otherwise, after one round, the distance between the two borders remains odd as both borders move toward each other by $\mathcal{A}_L$. 
    In the configuration reached $C''$, $u_c$ becomes a new border occupied by a crashed robot and non-crashed robots. 
    If there are only two occupied nodes in $C''$, by $\mathcal{A}_L$, the border robots move toward each other. 
    That is, in the next round, the border robots become at an even distance, and we are done. 
    If there are more than two occupied nodes in $C''$, by Lemmas \ref{lem:even-exists} and \ref{lem:same-side}, a configuration with a new target segment $[o,o']$ which includes one border node is reached. 
    Let us first consider the case in which $u_c \in [o,o']$. 
    If $|O_{[o,o']}|>1$, then after one round, the border robots become at an even distance, and we are done. 
    By contrast, if $|O_{[o,o']}|=1$, then we are done by Lemmas~\ref{lem:adjacent} and \ref{lem:Notadjacent}.
    Next, let us consider the case where $u_c \not\in [o,o']$. 
    Let $o_f$ be the closest occupied node of $u_c$. 
    Without loss of generality, $dist(o,u_c)<dist(o',u_c)$ holds.
    If $|O_{[o,o']}|>1$, then after one round, a configuration in which $[o', o_f]$ is the target segment and $u_c \in O_{[o', o_f]}$ is reached.  
    After one additional round, we are done. 
    Finally, if $|O_{[o,o']}|=1$, then robots on the nodes in $[o,o']$ move toward $u_c$, and we are done.
    
\end{itemize}
\end{enumerate}
From the cases above, we can deduce that the theorem holds.\qed
\end{proof}

From Lemmas \ref{theo:no-crash-ok} and \ref{theo2:crashed-ok}, we can deduce: 

\begin{theorem}
Starting from a non-edge-symmetric configuration $C$, algorithm $\mathcal{A}_L$ solves the SUIG problem on line-shaped networks without multiplicity detection in $O(\mathcal{D})$ rounds, where $\mathcal{D}$ denotes the distance between the two borders in $C$. \end{theorem}

\section{Concluding Remarks}
\label{sec:conclusion}

We initiated the research about stand-up indulgent rendezvous and gathering by oblivious mobile robots in the discrete model, studying the case of line-shaped networks. 
For both rendezvous and gathering cases, we characterized the initial configurations from which the problem is impossible to solve. 
In the case of rendezvous, a very simple algorithm solves all cases left open. 
In the case of gathering, we provide an algorithm that works when the starting configuration is not edge-symmetric. 
Our algorithms operate in the vanilla model without any additional hypotheses, and are asymptotically optimal with respect to the number of rounds to achieve rendezvous or gathering.

A number of open questions are raised by our work:
\begin{enumerate}
\item Is it possible to circumvent impossibility results in SSYNC using extra hypotheses (e.g., multiplicity detection)?
\item Is it possible to solve SUIR and SUIG in other topologies?
\end{enumerate}

\bibliographystyle{splncs04}
\bibliography{ref}

\end{document}